\pdfoutput=1
\documentclass[12pt]{article}

\usepackage{amssymb,amsfonts,amsmath,stmaryrd}
\usepackage{enumerate,float}
\usepackage{color}
\usepackage{mleftright}
\usepackage{tikz}
\usetikzlibrary{arrows,backgrounds}
\usepackage{ytableau}
\usepackage[vcentermath]{youngtab}
\usepackage{braket}
\usepackage{comment}
\usepackage{cases}
\usepackage{amsthm}
\DeclareMathAlphabet{\mathbbold}{U}{bbold}{m}{n}
\usepackage{dsfont}
\usepackage{mathtools}
\usepackage{doi}
\usepackage{comment}
\usepackage{authblk}

\theoremstyle{definition}

\def\be{\begin{eqnarray}}
\def\ee{\end{eqnarray}}

\DeclareMathOperator*{\tr}{\mathop{tr}}

\def\SP{{\rm SP}}

\definecolor{red}{rgb}{1,0,0}
\definecolor{orange}{rgb}{1,0.5,0}
\definecolor{violet}{rgb}{0.7,0,1}

\theoremstyle{plain}

\newtheorem*{prop*}{Proposition}

\usepackage[margin=3cm]{geometry} 

\title{\sffamily \textbf{Genus expansion of matrix models and $\hbar$ expansion of $B$KP hierarchy}}
\author{Yaroslav Drachov\thanks{\href{mailto:drachov.yai@phystech.edu}{\texttt{drachov.yai@phystech.edu}}} }
\author{Aleksandr Zhabin\thanks{\href{mailto:alexander.zhabin@yandex.ru}{\texttt{alexander.zhabin@yandex.ru}}}}
\affil{Moscow Institute of Physics and Technology,\authorcr 141701 Dolgoprudny, Russia}
\date{}
\usepackage[sorting=none, style=trad-alpha, url=false, backend=biber]{biblatex}
\usepackage{hyperref}
\hypersetup{hyperindex=true,breaklinks=true,linktoc=section}
\AtEveryBibitem{%
  \clearfield{note}%
}
\newcommand{\normord}[1]{
  \mathopen{:}\mathinner{#1}\mathclose{:}
}

\usepackage{extarrows}
\usepackage{blindtext,graphicx}
\usepackage[absolute]{textpos}
\begin{document}
%
%
\begin{textblock}{5}(13.59,2.65)
\begin{flushright}
MIPT/TH-05/23
\end{flushright}
\end{textblock}
\maketitle
%
%
%
%
%
%
%
%
%
\begin{abstract}
   We continue the investigation of the connection between the genus expansion of matrix models and the $\hbar$ expansion of integrable hierarchies started in \cite{APSZ20}. In this paper, we focus on the $B$KP hierarchy, which corresponds to the infinite-dimensional Lie algebra of type $B$. We consider the genus expansion of such important solutions as Br\'{e}zin-Gross-Witten (BGW) model, Kontsevich model, and generating functions for spin Hurwitz numbers with completed cycles. We show that these partition functions with inserted parameter $\hbar$, which controls the genus expansion, are solutions of the $\hbar$-$B$KP hierarchy with good quasi-classical behavior. $\hbar$-$B$KP language implies the algorithmic prescription for $\hbar$-deformation of the mentioned models in terms of hypergeometric $B$KP $\tau$-functions and gives insight into the similarities and differences between the models. Firstly, the insertion of $\hbar$ into the Kontsevich model is similar to the one in the BGW model, though the Kontsevich model seems to be a very specific example of hypergeometric $\tau$-function. Secondly, generating functions for spin Hurwitz numbers appear to possess a different prescription for genus expansion. This property of spin Hurwitz numbers is not the unique feature of $B$KP: already in the KP hierarchy, one can observe that generating functions for ordinary Hurwitz numbers with completed cycles are deformed differently from the standard matrix model examples.
\end{abstract}

\newpage
\tableofcontents
\section{Introduction}
The theory of matrix models has a long history and extensive 
applications in physics, mathematics, and many other fields
of knowledge. Some motivations, history, and applications
can be found, for example, in \cite{Mor94, Mir94, DFGZ95, EKR15}. This theory has
deep connections with integrable systems because matrix models form a large class of solutions ($\tau$-functions) of integrable hierarchies of KP/Toda type \cite{DJKM82, JM83}. In this paper we continue to develop deeper connections between the two, i.e. we study the similarity between genus expansion, coming from the matrix model side, and $\hbar$ expansion of integrable hierarchies on the other side.

In the literature, genus expansion is also known as
large $N$ expansion which goes back to 't Hooft \cite{tH74}. It comes from the perturbative expansion of matrix integrals as a sum over ribbon graphs. Each connected graph comes with a factor of $N^{2-2g}$, where $g$ is a genus of the surface where it
can be drawn and $N$ is the size of matrices in the ensemble.
That is, in the connected part of the partition function (free energy) one can distinguish contribution from surfaces of different genera
\begin{equation}
   F=\sum_{g=0}^{\infty} N^{2-2g} F_g
   \label{largeN}
.\end{equation}
Large $N$ factorization of correlators allows one to think of the large $N$ limit of matrix
models as a quasi-classical expansion. From different points of view, it can be also understood as a perturbative calculation of string amplitudes or as a WKB approximation.


In the presence of an external field in the action things get slightly more complicated \cite{KMMM95}. In general, the partition function depends both on the traces of the external matrix and on $N$. But quite often the models are defined to be independent on $N$, such that an explicit factor $N^{2-2g}$ does not appear. Nevertheless, genus expansion for such partition function still exists. Usually, they can be viewed as generating functions of some geometric quantities (for example, integrals over moduli spaces), which provide the genus expansion of the initial matrix integral. One can distinguish contributions from surfaces of different genera with the help of the formal parameter $\hbar$. The free energy is then
 \begin{equation}
   F^\hbar = \sum_{g=0}^{\infty} \hbar ^{2g}F_g
   \label{hexpansion}
.\end{equation}
Such an expansion generalizes \eqref{largeN} by taking $\hbar = 1/N$ and the transition from the partition function to free energy given by $F^{\hbar} = \hbar^{2} \log Z^{\hbar}$. The higher genus part of the expansion can be calculated from genus zero one- and two-point functions with the help of spectral curve topological recursion procedure \cite{AMM04, AMM05, CE06, EO08}.

From the side of integrable hierarchies, one can independently insert a formal parameter $\hbar$ into equations. At this moment $\hbar$ has nothing to do with the genus expansion parameter yet. Such a ``deformation'' of the KP hierarchy was considered in \cite{TT95} in order to study the dispersionless limit $\hbar \to 0$ of the hierarchy. The parameter is inserted in such a way that the standard KP hierarchy is restored at $\hbar=1$. At the first glance, such a deformation of the hierarchy seems trivial: if one looks only at the equations of the hierarchy, then the insertion of $\hbar$ is given just by the rescaling of ``times'': $t_{k} \to \frac{t_{k}}{\hbar}$. To satisfy such $\hbar$-deformed equations one needs simply to rescale times in the $\tau$-functions in the same way. However, one can notice that $\hbar$ can also be inserted, for example, in the Pl\"{u}cker coefficients so that the $\hbar$-KP equations are still satisfied. Thus, the single KP $\tau$-function can be $\hbar$-deformed to satisfy the $\hbar$-KP equations in many different ways.

Since $\hbar$ was inserted to study the dispersionless limit, it is natural to consider only those $\hbar$-deformations of $\tau$-functions for which the $\hbar \to 0$ limit exists. That is, the free energy should contain only non-negative powers of $\hbar$. And the trivial rescaling of times is not enough to satisfy this condition: $\hbar^{2} \log \tau \big( \frac{t}{\hbar} \big)$ would be very singular with respect to $\hbar$ as $\hbar \to 0$. That is, properly deformed $\tau$-functions should contain $\hbar$ non-trivially. In \cite{TT95} it was shown which insertion of $\hbar$ into the solutions is required in order to possess the good quasi-classical limit (see also \cite{APSZ20}). Even so, $\hbar$-deformation of the $\tau$-functions is still not unique under this additional condition. Nevertheless, we are interested only in the one particular $\hbar$-deformation of the $\tau$-function which coincides with its genus expansion.



In \cite{APSZ20} several particular examples of matrix model KP $\tau$-functions were considered. Among the examples were the Gaussian Hermitian model, Br\'{e}zin-Gross-Witten model, the Kontsevich model, and the generating function for simple Hurwitz numbers. All the models were deformed by the insertion of the parameter $\hbar$ responsible for the genus expansion. By definition of the genus expansion, they have good quasi-classical behavior of the form \eqref{hexpansion}. It was shown that, firstly, $\hbar$-deformed models are $\tau$-functions of $\hbar$-deformed KP hierarchy. That is, the genus expansion parameter from the matrix model side precisely coincides with the quasi-classical parameter on the integrable system side. Secondly, since almost all the mentioned partition functions (except for the Kontsevich model and BGW model in the Kontsevich phase) belong to the one family of hypergeometric $\tau$-functions, the general prescription for $\hbar$ insertion into these models was obtained.

Let us consider this algorithmic prescription for $\hbar$ deformation in more detail. The hypergeometric family of KP $\tau$-functions has the following expansion in Schur polynomials \cite{OSch01}:
\begin{equation}
    \tau(t) = \sum_{\lambda} f_{\lambda} S_{\lambda}(\beta) S_{\lambda}(t)
\end{equation}
where $f_{\lambda} = \prod_{(i,j) \in \lambda} f(i-j)$. Here $i,j$ are the coordinates of all the boxes in the Young diagram $\lambda$, and $f(n)$ is an arbitrary function. Choosing certain $f(n)$ gives a particular representative of the family. Let us perform the following insertion of parameter $\hbar$:
\begin{equation}\label{KPhyp}
    f(n) \to f(\hbar n), \hspace{1.5cm} \beta_{k} \to \frac{\beta_{k}}{\hbar}, \hspace{1.5cm} t_{k} \to \frac{t_{k}}{\hbar}
\end{equation}
Such a prescription appears to perform the genus expansion for the Gaussian Hermitian model, simple Hurwitz numbers, and the BGW model in the character phase. Naively, the genus expansion is very specific for a certain model. Nevertheless, it was shown that there exists a common pattern for it.

However, one can see that the Kontsevich model and BGW model in the Kontsevich phase do not fit into this elegant picture in the context of KP hierarchy because they do not belong to hypergeometric KP $\tau$-functions. The recent progress in the investigation of superintegrability in these models \cite{MM21, Ale21, Ale23} revealed that they belong to the hypergeometric family of $B$KP $\tau$-functions. The $B$KP hierarchy, which is similar to the KP hierarchy, has the parametrization of solutions by an infinite-dimensional algebra of type $B$. That is why one may seek for the natural $\hbar$-deformation of these models in the context of $B$KP hierarchy.

In this paper, firstly, we generalize the $\hbar$-deformation approach of \cite{APSZ20} to the $B$KP hierarchy in accordance with \cite{Tak93} and obtain the simple algorithmic prescription of $\hbar$-deformation for BGW and Kontsevich models in terms of family of hypergeometric $B$KP $\tau$-functions. That is, given that the family has the following expansion in $Q$-Schur polynomials \cite{Orl03} (see also section \ref{hypBKP}):
\begin{equation}
    \tau(t) = \sum_{\lambda \in \SP} r_{\lambda} Q_{\lambda} \mleft( \frac{\beta}{2} \mright) Q_{\lambda} \mleft( \frac{t}{2} \mright) 
\end{equation}
where $r_\lambda = \prod_{(i,j) \in \lambda} r(j)$, we obtain the following prescription for $\hbar$-deformation:
\begin{equation}\label{BKPhyp}
    r(n) \to r \mleft(\frac{1}{2} + \hbar \mleft(n - \frac{1}{2}\mright) \mright), \hspace{1.5cm} \beta_{k} \to \frac{\beta_{k}}{\hbar}, \hspace{1.5cm} t_{k} \to \frac{t_{k}}{\hbar}.
\end{equation}
We show that genus expansion of both the Kontsevich model and BGW model in the Kontsevich phase is governed by this prescription of $\hbar$-deformation, which implies that they are $\hbar$-$B$KP solutions. Also, we explain why a shift on $1/2$ is necessary in the $\hbar$-$B$KP case.

Secondly, we consider $\hbar$-deformation of another set of examples from $B$KP hypergeometric family, i.e. generating functions for spin Hurwitz numbers with completed cycles. These functions receive a lot of attention recently \cite{GKL21,AS21,MMN20,MMNO21,MMZ21,MMZ22} and they require special consideration. First of all, it is important to mention that the genus expansion of these models is done in the following way. The partition function is rewritten as a generating function of integrals over moduli spaces of curves with the help of cohomological representation motivated by Gromov-Witten theory \cite{GKL21}. This implies that several Hurwitz numbers, possibly counting surfaces of different genera, contribute to one such integral. The contribution of a certain genus $g$ to the partition function comes from the integrals over moduli spaces $\overline{\mathcal{M}}_{g,n}$ rather than from Hurwitz numbers themselves. Now, despite the fact that these $\tau$-functions belong to the hypergeometric family, their genus expansion is \textit{not} governed by \eqref{BKPhyp}. Instead, they have their own prescription for $\hbar$-deformation. While the non-deformed spin Hurwitz $\tau$-functions look like
\begin{equation}
    \tau_{\text{H}}(t) = \sum_{\lambda \in \SP} \exp \mleft( \sum_{r \in \mathbb{Z}_\text{odd}^+} u_r\mathbf{p}_{r}(\lambda) \mright) Q_{\lambda} \mleft( \frac{\beta}{2} \mright) Q_{\lambda} \mleft( \frac{t}{2} \mright)
\end{equation}
where Casimirs (or, $r$-th completed cycles) are $\mathbf{p}_{r}(\lambda) = \sum_{i=1}^{\ell(\lambda)} \lambda_i^r $, their $\hbar$-deformation prescription is given by
\begin{equation}\label{spinHurw_hbar}
    \mathbf{p}_{r}(\lambda) \to \hbar^{r-1} \mathbf{p}_{r}(\lambda), \hspace{1.5cm} \beta_{k} \to \frac{\beta_{k}}{\hbar}, \hspace{1.5cm} t_{k} \to \frac{t_{k}}{\hbar}.
\end{equation}
To compare this to the deformation \eqref{BKPhyp}, one needs to rewrite $ e^{u_{r} \mathbf{p}_{r}(\lambda)} = \prod_{(i,j)\in\lambda} e^{u_{r} \mathbf{p}_{r}(j)}$, which implies rewriting the $\hbar$-deformation \eqref{spinHurw_hbar} as:
\begin{equation}\label{spinHurw_hbar_as_hyp}
    \mathbf{p}_{r}(n) \to \hbar^{r-1} \mathbf{p}_{r}(n), \hspace{1.5cm} \beta_{k} \to \frac{\beta_{k}}{\hbar}, \hspace{1.5cm} t_{k} \to \frac{t_{k}}{\hbar}.
\end{equation}
We show that $\mathbf{p}_{r}(n) = (n^{r} - (n-1)^{r})$, thus, prescription \eqref{spinHurw_hbar_as_hyp} is different from \eqref{BKPhyp}. Nevertheless, we show that $\hbar$-deformed spin Hurwitz partition functions are still solutions of $\hbar$-$B$KP with good quasi-classical limit \eqref{hexpansion}.

Such a discrepancy between the prescription of $\hbar$-deformation tends us to revisit the KP case \eqref{KPhyp} and consider generating functions for ordinary Hurwitz numbers with completed cycles, which were not mentioned in \cite{APSZ20}. The proper $\hbar$-deformation of such partition functions was considered, for example, in \cite{BDKS20} and, indeed, we observe the same picture. While all the ``simple'' examples of \cite{APSZ20} are deformed by \eqref{KPhyp}, ordinary Hurwitz numbers with completed cycles have their own $\hbar$-deformation prescription. The only standing out example is the generating function for simple Hurwitz numbers for which both prescriptions coincide. That is, this discrepancy in $\hbar$-deformation of Hurwitz numbers appears both in KP and $B$KP cases.

The reason for the difference in the $\hbar$-deformation lies in the form of the function $f(n)$ (or $r(n)$). Starting from the KP/$B$KP point of view the difference is not visible. However, the theory of \cite{BDKS20} gives a recipe for deformation starting from the spectral curve. And for spectral curves corresponding to Hurwitz numbers, one can see the more complicated way of deformation than for the other simple examples. The function $f(n)$ (or $r(n)$) appears to be tightly connected with the spectral curve data, which implies the difference in the genus expansion. Thus, the deformation recipe is sensitive to either function $f(n)$ (or $r(n)$) is polynomial or exponential, as it is in BGW and spin Hurwitz examples respectively. From this point of view, the Kontsevich model stands aside because of the very special form of $r(n)$, depending on $n \bmod 3$ (see section \ref{ss:kont}). Even so, we observe that the Kontsevich model indeed deforms as simply as the BGW model. And this fact suggests that there exists some generalization of \cite{BDKS20} on the $B$KP case which should include the Kontsevich model as a special case.

The paper is organized as follows. In section \ref{ss:2} we introduce the notations which we use throughout the paper, such as $Q$-Schur polynomials, algebra $\widehat{\mathfrak{go}\mleft( \infty \mright)} $ and neutral fermions. Section \ref{ss:3} is devoted to the $B$KP hierarchy and, in particular, to the important class of hypergeometric solutions which arises in examples. In section \ref{ss:4} we consider $\hbar $-formulation of $B$KP hierarchy. In section \ref{ss:ehbkp} we explicitly introduce such solutions of $\hbar $-$B$KP as BGW and Kontsevich models. For each model we explain the insertion of the parameter and prove that they are, indeed, solutions of $\hbar $-$B$KP. Moreover, we show the common pattern of insertion of $\hbar$ into these models in terms of hypergeometric $\tau$-functions. Finally, in section \ref{s:shn} we consider generating functions of both ordinary and spin Hurwitz numbers. Firstly, we revisit the KP case and discuss the difference in $\hbar$-deformation of these functions with the other examples. Then we switch to the $B$KP case and observe the same feature for the spin Hurwitz numbers. 

\section{\texorpdfstring{$B$}{B}KP boson-fermion correspondence}
\label{ss:2}
In this section, we introduce the language that we are using throughout the paper: we briefly review the most important facts about $Q$-Schur polynomials, infinite-dimensional Lie algebra of type $B$, and boson-fermion correspondence in the context of $B$KP integrable hierarchy. The latter allows one to identify bosonic and fermionic descriptions and to use the more convenient one.
Extensive information about $Q$-Schur polynomials 
can be found in \cite{Mac98}, here we summarize only what we do need.
\subsection{\texorpdfstring{$Q$}{Q}-Schur polynomials}
Firstly, let us introduce the $Q$-Schur polynomials. Let us consider an ordered set of 
non-negative integers $\lambda_1 \ge \lambda_2\ge \cdots
\ge \lambda_\ell \ge 0$. We denote this set by $\lambda
= \left[ \lambda_1,\lambda_2 ,\ldots,\lambda_\ell \right] $ 
and call it a Young diagram. Each Young diagram
corresponds to a partition of
an integer $|\lambda|= \lambda_1 +\lambda_2+\cdots +\lambda_\ell$ onto $\ell(\lambda)$ non-zero parts $\lambda_i$.
Graphical representation  of Young diagrams is a finite
collection of boxes, arranged in left-justified rows,
with the length of each row equal to $\lambda_1,\lambda_2,\ldots,\lambda_\ell$. For example, the diagram  $\left[ 5, 3, 2 \right] $ has the following graphical representation:
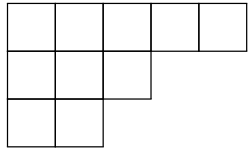
\begin{center}
\begin{ytableau}
	\, & & & &\\
	& &   \\
	&  
	\end{ytableau}
\end{center}	

Consider an  infinite set  of variables with odd indices $ t = \left\{ t_1, t_3,\ldots	\right\} $.
$Q$-Schur polynomials $Q_\lambda(t)$ are labeled
by Young diagrams and are defined via the pfaffian formula:
\begin{equation}
	Q_\lambda(t) := 2^{-\ell(\lambda) /2} \operatorname{Pf}
	M_\lambda(t) \equiv 2^{-\ell(\lambda)/ 2} \sqrt{\det 
	M_\lambda  (t)} 
	\label{eq:Ql}
,\end{equation}
where antisymmetric matrix $M_\lambda (t)$ is defined as
\begin{equation}
	\left( M_\lambda (t) \right) _{i,j} :=
	P_{\lambda_i,\lambda_j} (t)
\end{equation} 
for even $\ell(\lambda)$; if  $\ell(\lambda)$ is odd, add
exactly one line $P_{0,\lambda_i}(t)$ (as if one adds
zero length line to the Young diagram  $\lambda$). In turn, polynomials $P_{n,m}(t)$ are evaluated with the help of the following generating
function:
\begin{equation}
\sum_{n,m=0}^{\infty} P_{n,m} z_1^m z_2^m=
\left[ \exp \left( 2 \sum_{k=0}^{\infty} t_{2k+1}
\left( z_1^{2k+1}+z_2^{2k+1} \right) \right) -1 \right] 
\frac{z_1-z_2}{z_1+z_2}
	\label{}
.\end{equation}
This definition implies that the $Q$-Schur polynomials vanish if $\lambda$ contains two lines of equal
length. The diagrams that do not have lines of equal length are called strict partitions (SP).

To be precise, let us introduce the first few non-trivial examples of $Q$-Schur polynomials:
\ytableausetup{boxsize=0.2em}
\begin{equation}
	\begin{gathered}
	Q_\emptyset(t)=1,	\\
	Q_{\ydiagram{1}}(t)=\sqrt{2} t_1,\\
	Q_{\ydiagram{2}}(t)=\sqrt{2}t_1^2,\\
	Q_{\ydiagram{3}}(t)=\sqrt{2}\left( \frac{2t_1^3}{3}+t_3 \right) ,\qquad Q_{\ydiagram{2,1}}(t)=2\left( \frac{t_1^3}{3}-t_3 \right),  \\
	Q_{\ydiagram{4}}=\sqrt{2} \left( \frac{t_1^4}{3}+2t_1t_3 \right) ,\qquad
	Q_{\ydiagram{3,1}}=2\left( \frac{t_{1}^4}{3}-t_1t_3 \right). 
	\end{gathered}
	\label{}
\end{equation}
%
%
\subsection{Fock space, neutral fermions and \texorpdfstring{$\widehat{\mathfrak{go}(\infty)}$}{go(∞)} algebra}
Again, here we summarize only the necessary results. The original approach to neutral fermion Fock space is described in great detail in \cite{DKM81,DJKM82,You89,vdLeu95,Orl03}. In our notations, we follow a more recent summary of \cite{Ale23}.
There is a natural way to describe solutions of $B$KP hierarchy
in terms of neutral fermions. Firstly, let us
introduce an infinite-dimensional Clifford algebra with generators
$\phi_k,\ k \in \mathbb{Z}$ and
commutation relations:
\begin{equation}
   \left\{ \phi_k,\phi_m \right\} =(-1)^k \delta_{k+m,0}
  \label{eq:1}
.\end{equation} 
Note that $\phi_0^2 = 1 /2$. One can make neutral fermions from
regular free fermions
\begin{equation}
   \phi_k= \frac{\psi_k+(-1)^k \psi_{-k}^{*}}{\sqrt{2} }
   \label{}
,\end{equation}
\begin{equation}
   \mleft\{ \psi_k,\psi^*_{l} \mright\} =\delta_{il},\qquad
   \mleft\{ \psi_k,\psi_l \mright\} =0,\qquad
   \mleft\{ \psi_k^*,\psi_l^* \mright\} =0
   \label{}
.\end{equation}
Let us introduce generating series for neutral fermions:
\begin{equation}
   \phi(z)= \sum_{k \in \mathbb{Z}}^{} \phi_k z^k
   \label{eq:fermfield}
.\end{equation} 
Neutral fermion Fock space $\mathcal{F}$ (and its dual $\mathcal{F}^{*}$) is defined by
the action of Clifford
algebra on the vacuum vector $\ket{0}$ (and respectively corresponding 
co-vacuum vector $\bra{0}$):
\begin{equation}
\phi_k \ket{0}=0,\qquad \bra{0}\phi_{-k}=0,\qquad k <0
,\end{equation}
and the elements $\phi_{k_1} \phi_{k_2} \cdots \phi_{k_m}\ket{0}$ with
$k_1 >k_2 > \ldots > k_m \ge 0$ form basis in $\mathcal{F}$. The linear space $\mathcal{F}$ splits into two subspaces
\begin{equation}
   \mathcal{F}= \mathcal{F}^0 \oplus \mathcal{F}^1
\end{equation}
where $\mathcal{F}^0$ and $\mathcal{F}^1$ denote the subspaces with even and odd number of generators respectively.
If we consider only the space $\mathcal{F}^0$, the basis in it can be
labelled by strict partitions $\lambda \in \mathrm{SP}$
in the following way:
\begin{equation}
   \ket{\lambda}= \begin{cases}
      \phi_{\lambda_1}\phi_{\lambda_2}\cdots \phi_{\lambda_{\ell(\lambda)}}\ket{0},& \ell(\lambda)=0 \bmod 2,\\
      \sqrt{2} \phi_{\lambda_1}\phi_{\lambda_2} \cdots 
      \phi_{\lambda_{\ell(\lambda)}} \phi_{0} \ket{0}, & \ell(\lambda)=1 \bmod 2.
   \end{cases}
   \label{eq:l}
\end{equation}

With the help of commutation relations \eqref{eq:1} it is easy to see that
\begin{equation}
\braket{0 | \phi_k \phi_m | 0}= \delta_{k+m,0}
\eta (m)
,\end{equation} 
where
\begin{equation}
   \eta(m)= \begin{cases}
      0,& m<0,\\
      1 /2,& m=0,\\
      (-1)^m,& m>0.
   \end{cases}
\end{equation}

We denote the normal ordering of fermionic operators as $\normord{\left( \ldots \right) }$,
which means that all  annihilation operators are moved to the right and all creation operators to the left, with respect  to $(-1)$ with each transposition of fermions. For example, $\normord{\phi_{-2}\phi_{1}}=-\phi_1 \phi_{-2}$. Note that it is not the same as the transposition of fermions with the  help of commutation relations \eqref{eq:1}.

%
Let's consider a Lie
algebra of matrices $\mathfrak{go}(\infty)$: each matrix $A \in \mathfrak{go}(\infty)$ 
is an infinite matrix, with two additional requirements:
\begin{itemize}
\item 
only
finitely many diagonals are non-zero, $A_{ij} = 0$ for  $\left| i-j \right| \gg 0$ is satisfied,
\item $A$ is skew-symmetric, $A_{ij}=- A_{ji}$.
\end{itemize}

Let us introduce matrices $E_{ij}$ that have 1 on $i,\ j$'s
place and 0 everywhere else, i.\:e. $\left( E_{ij} \right) _{kl}=
\delta_{ik}\delta_{jl}$. Then the standard basis for the algebra $\mathfrak{go}(\infty)$  consists of matrices 
\begin{equation}
   F_{km}=\left( -1 \right) ^{m} E_{km} -\left( -1 \right) ^{k}
   E_{-m,-k}
.\end{equation} 
It is easy to show that $F_{ij}$ satisfy standard commutation relations
\begin{equation}
   \left[ F_{ij},F_{kl} \right] = (-1)^j \delta_{jk}F_{il}-
   \left( -1 \right) ^i \delta_{i,-k} F_{-j,l}+\left( -1 \right) ^j
   \delta_{j,-l} F_{k,-i} -\left( -1 \right) ^i \delta_{i,l} F_{jk}
.\end{equation} 
Now let us consider $\mathfrak{go}\left( \infty \right) $'s central extension, $\widehat{\mathfrak{go}(\infty)}$. 
As a linear space, it is $\mathfrak{go}(\infty) \oplus  \mathbb{C}c$. The commutator of two arbitrary
elements $A,B \in  \widehat{\mathfrak{go}(\infty)}$ is given by
\begin{equation}
   \left[ A,B \right] = AB-BA+\alpha\left( A,B \right) c
   \label{eq:comm}
\end{equation}
where $\alpha(A,B)$ is linear in each variable and therefore can be
defined on the basis elements $F_{ij}$:
\begin{equation}
   \alpha\left( F_{ij},F_{kl} \right) =  \left( 
   \delta_{kj}\delta_{il}-\delta_{ik}\delta_{jl}\right) \left( (-1)^i
\eta\left( j \right) -\left( -1 \right) ^j \eta(i)\right) \mathbbold{1} 
,\end{equation} 
where $\mathds{1} $ is the identity element. Now the $\widehat{\mathfrak{go}(\infty)}$ representation on space $\mathcal{F}$ is, in terms of basis elements:
\begin{equation}
   r(F_{ij})= \normord{\phi_i \phi_j}
\end{equation} 
It is straightforward  to check that $r$  is indeed a representation:
commutation relations \eqref{eq:comm}
for $r(A)$, $r(B)$ are preserved. The central charge $c$ 
in this representation is equal to 1.

Now let us consider operators
\begin{equation}
   H_n= \frac{1}{2}\sum_{k \in \mathbb{Z}}^{}(-1)^{k+1} \normord{\phi_k \phi_{-k-n}},\qquad
   n \in \mathbb{Z}_{\text{odd}}
\end{equation} 
which satisfy the following commutation relations
\begin{equation}
\left[ H_n,H_m \right] =\frac{n}{2}\delta_{n,-m}
\end{equation} 
and thus generate the Heisenberg subalgebra $\mathcal{A} \in \widehat{\mathfrak{go}(\infty)}$.

It appears that one can construct an isomorphic representation of $\widehat{\mathfrak{go}(\infty)}$ in bosonic Fock space with the help of maps $\Phi^i\colon \mathcal{F}^i \to  \mathcal{B}^i=\mathbb{C} \llbracket t_1,t_3,t_5,\ldots \rrbracket
$ for $i=0,1$, which are homomorphisms of representations \cite{You89}. Under this homomorphism
the operators $H_k$ map into multiplication and differentiation
w.\:r.\:t. times
\begin{equation}
\left\{
\begin{aligned}
&H_k \to \frac{\partial }{\partial t_k} ,\\
&H_{-k} \to  \frac{k}{2} t_k,\\
&H_0 \to 0 ,
\end{aligned}
\right. \qquad k \in \mathbb{Z}^{+}_\text{odd}.
\end{equation} 
Maps $\Phi^i$ provide the $B$KP variant of boson-fermion correspondence. In what follows we consider $B$KP hierarchy, thus we are interested only in the action of $GO(\infty)$ on vacuum vector $\ket{0}$. Group $GO(\infty)$ is a standard exponential map from the algebra $\widehat{go(\infty)}$. Maps $\Phi^i$, in this case, can be explicitly written as vacuum expectation value:
\begin{equation}
   \Phi^i\left( G \ket{0} \right) =  \begin{cases}
      \braket{1 | e^{H(t)} G | 0}, & G\ket{0} \in \mathcal{F}^1,\\
      \braket{0 | e^{H(t)}G | 0}, & G \ket{0} \in \mathcal{F}^0,
   \end{cases}
\end{equation} 
where $\bra{1} = \sqrt{2} \bra{0}\phi_0$ and
\begin{equation}
   H(t)= \sum_{k \in \mathbb{Z}^+_{\text{odd}}}^{} t_k H_k
.\end{equation} 
The explicit isomorphism between spaces $\mathcal{F}^{0}$ and $\mathcal{B}^{0}$ is given by
\begin{equation}
   \Phi^0 \mleft( \ket{\lambda} \mright) =Q_\lambda\mleft( \frac{t}{2} \mright) 
   \label{}
\end{equation}
That is, the boson-fermion correspondence provides an explicit realization for the states
\eqref{eq:l} in terms of Q-Schur polynomials.

\section{\texorpdfstring{$B$}{B}KP hierarchy}
\label{ss:3}
In this section, we briefly review the main facts about $B$KP equations and their solutions. For a detailed explanation see the original papers \cite{DJKM82,JM83,You89} or \cite{Orl03}. $B$KP hierarchy is an infinite set of non-linear differential equations with the first equation given by
\begin{multline}
-60 \left(\frac{\partial ^2F}{\partial
   t_1^2}\right)^3-30 \frac{\partial ^4F}{\partial
   t_1^4} \frac{\partial ^2F}{\partial
   t_1^2}+30 \frac{\partial ^2F}{\partial t_1\, \partial t_3}
   \frac{\partial ^2F}{\partial t_1^2}-\frac{\partial
   ^6F}{\partial t_1^6}\\+5 \frac{\partial ^2F}{\partial
   t_3^2}-9 \frac{\partial ^2F}{\partial t_1\, \partial t_5}+5
   \frac{\partial ^4F}{\partial t_1^3\, \partial t_3}=0	
	\label{}
.\end{multline}
It is more common to work with $\tau$-function $\tau(t)= \exp F(t)$. We assume that $\tau(t)$ is at least a formal power series
in times $t_k$, and maybe it is even a convergent series. The entire
set of equations of the hierarchy can be written in terms of $\tau$-function
using Hirota bilinear identity
\begin{equation}
\frac{1}{2\pi i } \oint
e^{\xi\left( t-t',\,k \right) }\tau\left( t-2 \left[ k^{-1} \right]  \right) \tau\left( t'+2 \left[ k^{-1} \right]  \right) \frac{d k}{k}=
\tau\left( t \right) \tau
(t')
\label{eq:hir}
,\end{equation}
where
\begin{equation}
	t\pm \left[ k^{-1} \right] 
	=
	\left\{ t_1\pm k^{-1},\,t_2 \pm \frac{1}{2}
	k^{-2},\,t_3\pm \frac{1}{3}k^{-3},\ldots\right\} 
\end{equation}
and
\begin{equation}
	\xi\left( t,\,k \right)=
	\sum_{j \in \mathbb{Z}_{\text{odd}}^+}^{} t_j k^j
\end{equation}
Contour integration $\oint \frac{dk}{2\pi i}$ here means that we expand integrand at the point $k=\infty$ and take the coefficient of $k^{-1}$.

With the change of variables $t_j = t_j + \epsilon _j  $, $t_j' = t_j - \epsilon _j$ one can rewrite bilinear identity in terms of Hirota derivatives
\begin{equation}
\oint \frac{d k}{2\pi i k}
e^{2\xi\left( \epsilon ,\,k \right)}  \exp \mleft( \sum_{j=1}^{\infty} \mleft( \epsilon _j -\frac{2}{ j k^j} \mright) D_j \mright) \tau \cdot \tau =
\tau\left( t \right) \tau
(t')
\end{equation}
Expanding integrand in powers of $\epsilon _j$ and taking the coefficient of $k^{-1}$ one obtains $B$KP equations.

On the one hand, all formal power series solutions of $B$KP hierarchy can be decomposed over the basis of $Q$-Schur polynomials
\begin{equation}
   \tau\mleft(t\mright)= \sum_{\lambda  \in \mathrm{SP}}^{} C_\lambda Q_\lambda\mleft(\frac{t}{2}\mright)
   \label{eq:bosonic}
.\end{equation} 
Function written as a formal sum over $Q$-Schur polynomials is a $B$KP solution if and only if coefficients $C_\lambda$ satisfy the $B$KP Pl\"{u}cker relations:
\begin{multline}
   C_{\mleft[ \alpha_1,\ldots\alpha_k \mright] }C_{\mleft[ \alpha_1,\ldots,\alpha_k,\beta_1,\beta_2,\beta_3,\beta_4 \mright] }-
   C_{\mleft[ \alpha_1,\ldots,\alpha_k,\beta_1,\beta_2 \mright] }
   C_{\mleft[ \alpha_1,\ldots,\alpha_k,\beta_3,\beta_4 \mright] }
\\+
   C_{\mleft[ \alpha_1,\ldots,\alpha_k,\beta_1,\beta_3 \mright] }
   C_{\mleft[ \alpha_1,\ldots,\alpha_k,\beta_2,\beta_4 \mright] }
-
   C_{\mleft[ \alpha_1,\ldots,\alpha_k,\beta_1,\beta_4 \mright] }
   C_{\mleft[ \alpha_1,\ldots,\alpha_k,\beta_2,\beta_3 \mright] }=0
   \label{eq:plucker_general}
.\end{multline}
The first non-trivial relation is
\ytableausetup{boxsize=0.2em}
\begin{equation}
	C_{\varnothing}C_{\ydiagram{3,2,1}}-C_{\ydiagram{1}}C_{\ydiagram{3,2}}+C_{\ydiagram{2}}C_{\ydiagram{3,1}}-C_{\ydiagram{3}}C_{\ydiagram{2,1}}=0
	\label{eq:plucker}
.\end{equation}
We call $\tau$-functions of the form  \eqref{eq:bosonic} as $\tau$-functions in \emph{bosonic representation}.

On the other hand, $\tau$-function is an image under the boson-fermion correspondence of a point on the orbit of the vacuum $\ket{0}$ under the action of some element $G$ of $GO(\infty)$:
\begin{equation}
   \tau(t)= \braket{0 | e^{H(t)}G | 0}
\end{equation}
Such $\tau$-functions are called $\tau$-functions in \emph{fermionic representation}. In what follows we use both representations for $B$KP $\tau$-functions.
%
%
\subsection{Hypergeometric \texorpdfstring{$\tau$}--functions}\label{hypBKP}
In this paper, we are interested in  a subset of $B$KP  $\tau$-functions of hypergeometric type, or simply hypergeometric $\tau$-functions. This relatively simple set of $B$KP solutions contains surprisingly many physical examples. It was first introduced in \cite{Orl03}. In fermionic representation these $\tau$-functions have the form:
\begin{equation}
   \begin{gathered}
   \tau(t)= \braket{0 | e^{H(t)}e^{B(\beta)} | 0},\\
   B(\beta)= \sum_{k \in \mathbb{Z}^+_{\text{odd}}}^{} \beta_k B_k,\qquad
   B_k= -\frac{1}{2}\sum_{n \in \mathbb{Z}}^{} r(n)r(n-1)\cdots
   r(n-k+1) \phi_n \phi_{k-n}
   \label{eq:bkphyp}
   \end{gathered}
\end{equation} 
where function $r(n)$ has to satisfy
\begin{equation}
   r(n)=r(1-n)
\end{equation}
and $\beta= \left\{ \beta_1,\,\beta_2,\ldots \right\} $ is an arbitrary set of parameters. Matrix $B_k$ has non-zero elements of a specific form on the $k$- and $(-k)$-th diagonals. It is easy to obtain another form of the matrix $B_k$ which is more convenient in some cases:
\begin{equation}
   B_k= -\frac{1}{4\pi i}\oint \frac{dz}{ z} \phi\mleft(-z\mright) \mleft( \frac{1}{z}r(D) \mright) ^k \phi(z)
   \label{eq:ak}
\end{equation} 
where $D= z \partial/\partial z$, therefore $r(D)z^n= r(n) z^n$. Using the explicit form of the fermionic fields \eqref{eq:fermfield} one can obtain \eqref{eq:ak}.

Bosonic representation of hypergeometric $\tau$-functions requires
a notion of a $B$KP-content $c(w)$ of a box $w$ of Young diagram $\lambda$:
\begin{equation}
   c(w)=j,\qquad 1\le i \le \ell(\lambda),\qquad 1\le j \le  \lambda_i
   \label{eq:cont}
.\end{equation} 
For example, boxes of the diagram $\left[ 5,\,3,\,2 \right] $ 
have the following contents:
\ytableausetup{boxsize=1.5em}
\begin{center}
\begin{ytableau}
   1&2&3&4&5\\
	1&2&3\\
	1&2
\end{ytableau}
\end{center}
Hypergeometric $\tau$-functions in bosonic representation are
given by
\begin{equation}
   \tau(t)= \sum_{\lambda \in  \mathrm{SP}} r_\lambda Q_\lambda\left(\frac{\beta}{2}\right) Q_\lambda\left(\frac{t}{2}\right)
,\end{equation} 
where $Q_\lambda\left(\beta/2\right)$ is a $Q$-Schur polynomial of
variables $\beta_k /2$ and
\begin{equation}
   r_\lambda=  \prod_{w \in \lambda}^{} r(c(w))   
   \label{eq:rl}
.\end{equation}

\section{\texorpdfstring{$\hbar $}{ħ}-formulation of \texorpdfstring{$B$}{B}KP hierarchy}
\label{ss:4}
In this section, we introduce a formal parameter $\hbar$ in the $B$KP hierarchy (we shortly call it $\hbar$-$B$KP). The idea to study $\hbar$-$B$KP was first formulated in \cite{Tak93} in order to investigate the dispersionless limit of the hierarchy. In our work we follow the way analogous to $\hbar $-formulation of the KP hierarchy in \cite{TT95,TT99,NZ16}.

Let us define $\hbar$-$B$KP hierarchy as an ordinary $B$KP hierarchy with rescaled variables $t_k \to t_k /\hbar $, and redefined $F= \hbar ^2 \log \tau$. The first equation of the hierarchy is then of the form
\begin{multline}
-60 \left(\frac{\partial ^2F}{\partial
   t_1^2}\right)^3-30\hbar ^2 \frac{\partial ^4F}{\partial
   t_1^4} \frac{\partial ^2F}{\partial
   t_1^2}+30 \frac{\partial ^2F}{\partial t_1\, \partial t_3}
   \frac{\partial ^2F}{\partial t_1^2}-\hbar ^4\frac{\partial
   ^6F}{\partial t_1^6}\\+5 \frac{\partial ^2F}{\partial
   t_3^2}-9 \frac{\partial ^2F}{\partial t_1\, \partial t_5}+5
   \hbar ^2\frac{\partial ^4F}{\partial t_1^3\, \partial t_3}=0	
	\label{}
.\end{multline}
As one can see, it is not a ``deformation'' in any sense, it is just simply a rescaled original $B$KP hierarchy that allows $F$-functions to depend on arbitrary powers of formal parameter  $\hbar$. One can obtain $\tau$-functions of $\hbar $-$B$KP from  $B$KP ones with the change of variables $t_k \to t_k /\hbar $.

Non-triviality appears when we restrict free energy $F$ not to be singular in $\hbar$. That allows to perform dispersionless limit $\hbar \to 0$.
This requirement imposes restrictions on $F$:
parameter $\hbar $ should be inserted into the
logarithm of $\tau$-function and equations ``properly'', that is, $F$-function must not contain any negative powers of $\hbar $:
\begin{equation}
   F= \sum_{k=0}^{\infty} \hbar ^{k}F^{(k)}(t)
   \label{}
.\end{equation}
The dispersionless
free energy should satisfy the dispersionless hierarchy with the first equation given by
\begin{equation}
-60 \left(\frac{\partial ^2F^{(0)}}{\partial
t_1^2}\right)^3+30 \frac{\partial ^2F^{(0)}}{\partial t_1\, \partial t_3}
\frac{\partial ^2F^{(0)}}{\partial t_1^2}+5 \frac{\partial ^2F^{(0)}}{\partial
	 t_3^2}-9 \frac{\partial ^2F^{(0)}}{\partial t_1\, \partial t_5}=0	
	\label{}
.\end{equation}

The non-triviality in $\hbar$-dependence of $F$-function can be obtained by inserting $\hbar $ in the Pl\"{u}cker
coefficients as well as a rescaling of times:
\begin{equation}
   \tau^\hbar \mleft(t\mright)= \sum_{\lambda}^{} C_\lambda^\hbar 
   Q_\lambda \mleft( \frac{t}{\hbar } \mright) 
   \label{eq:th}
.\end{equation}
\begin{prop*}[$\hbar $-$B$KP solution criterion]
$\tau$-function of the form \eqref{eq:th} satisfies $\hbar $-$B$KP
 equations if and only if coefficients $C_\lambda^\hbar $ 
 satisfy classical Pl\"{u}cker identities.
\end{prop*}
\begin{proof}
If $\tau^\hbar (t)$ solves  $\hbar $-$B$KP then $\tau^\hbar (t\hbar )$ solves $B$KP, therefore $C_\lambda^\hbar $ 
 satisfy the classical Pl\"{u}cker relations. If $C_\lambda^\hbar $ 
 satisfy the classical Pl\"{u}cker relations then by the same logic $\tau^\hbar (t)$ solves $\hbar $-$B$KP. 
\end{proof}

In the following sections, we show that for known solutions of the $B$KP hierarchy, their genus expansion (given by non-trivial insertion of $\hbar $ in $C_\lambda$), however, solves $\hbar $-$B$KP. In this sense, we truly deform our original $B$KP  $\tau$-functions with the formal parameter $\hbar $ and obtain $\hbar $-$B$KP $\tau$-functions.

\section{Examples of \texorpdfstring{$\hbar$}{ħ}-\texorpdfstring{$B$}{B}KP solutions}
\label{ss:ehbkp}
In this section, we discuss separately two solutions of 
$B$KP: Br\'{e}zin-Gross-Witten model and Kontsevich model. For each of these $\tau$-functions we show
two things:
 \begin{itemize}
\item genus expanded $\tau$-function satisfies the $\hbar $-$B$KP hierarchy,
  \item genus expansion of both $\tau$-functions can be obtained by following one simple prescription.
\end{itemize}
Even though this prescription works for the Kontsevich and BGW models, it changes for the generating functions for spin Hurwitz numbers. For now, we leave the discussion of the difference until the next section.

\subsection{Br\'{e}zin-Gross-Witten model}
Firstly, let us discuss the Br\'{e}zin-Gross-Witten model and its genus expansion. This model has two phases and both of them were already considered in \cite{APSZ20} in the context of the KP hierarchy. It was shown that the $\hbar$-deformation of the model (which reveals the genus expansion) is a solution of the $\hbar$-KP hierarchy in both phases. While the genus expansion of the character phase is governed by the general prescription for $\hbar$-deformation of hypergeometric KP $\tau$-functions \eqref{KPhyp}, the Kontsevich phase does not fit into this family and therefore into this prescription of deformation. Here we treat the Kontsevich phase of the model as the simplest non-trivial hypergeometric solution of the BKP hierarchy and develop a prescription for $\hbar$-deformation of such $\tau$-functions. In particular, we show that in this phase the deformed BGW model is a solution of the $\hbar$-$B$KP hierarchy.

\subsubsection{Classical BGW model}
BGW model was first introduced as a partition function of 2D lattice gauge theories \cite{GW80,BG80}:
\begin{equation}
   Z_{\text{BGW}}\left( J,J^\dagger \right) =
    \int DU e^ {\operatorname{tr}
   \left( J^\dagger U+J U^\dagger \right) } 
\end{equation} 
where the integration is over $N \times N$ unitary matrices  
with the Haar measure $DU$ of the unitary group $U(N)$, normalized by $\int DU=1$, and $N\to \infty$.

Since the Haar measure is invariant with respect to the group action,
$Z_\text{BGW}\left( J, J^\dagger \right) $ depends only on $N$ 
parameters, the eigenvalues of the matrix $JJ^\dagger$.
Depending on the choice of variable $t_k$, there are two
phases \cite{MMS96}:
\begin{equation}
   t_k= \frac{1}{k} \tr \left( J J^\dagger \right) ^k \text{ --- character phase (} J\to 0\text{)}
   \label{}
.\end{equation}	
\begin{equation}
   t_k= - \frac{1}{2k-1} \tr \left( J J^\dagger \right) ^{-k+1 /2}\text{ --- Kontsevich phase } \mleft(\frac{1}{J}\to 0\mright)
   \label{}
.\end{equation}
We focus only on the Kontsevich phase since the BGW model in this phase is a solution to the BKP hierarchy and has a simple expansion in $Q$-Schur polynomials \cite{Ale23}
\begin{equation}
   \tau_\text{BGW} \mleft( \frac{t}{2} \mright) =
   \sum_{\lambda \in \mathrm{SP}}^{} r_\lambda Q_\lambda\mleft( \frac{2\delta_{k,1}}{2} \mright) Q_\lambda\mleft( \frac{t}{2} \mright)  ,\qquad
   r(n) = \frac{(2n-1)^2}{16}
   \label{eq:rn}
\end{equation} 
So we see that the partition function of the BGW model in the properly normalized times is a hypergeometric $\tau$-function of the $B$KP hierarchy  with $\beta_k= 2 \delta_{k,1}$.

In fermionic formalism BGW $\tau$-function looks like
 \begin{equation}
   \tau_{\text{BGW}}\mleft(\frac{t}{2}\mright)=
   \braket{0 | \exp H\mleft( t \mright)\exp \mleft( 
   - \frac{1}{\pi i}\oint \frac{dz}{z}  \phi(-z) 
\frac{(2D-1)^2}{16z}\phi(z) \mright)  | 0}
   \label{}
.\end{equation}
\subsubsection{BGW as a solution of \texorpdfstring{$\hbar $}{ħ}-\texorpdfstring{$B$}{B}KP}
The BGW model is a generating function for intersection
numbers of $\Theta$-classes and $\psi$-classes on compactified moduli spaces $\overline{\mathcal{M}}_{g,n}$ of complex curves of genus $g$ with $n$ marked points \cite{Nor17}. These intersection numbers
\begin{equation}
   \int_{\overline{\mathcal{M}}_{g,n}}\Theta_{g,n}\psi_1^{m_1} \psi_{2}^{m_2}
   \cdots \psi_n^{m_n} = \left< \tau_{m_1} \tau_{m_2} \cdots
   \tau_{m_n}\right>^\Theta
   \label{}
\end{equation}
are rational numbers, which are not equal to zero only if
\begin{equation}
   \sum_{i=1}^{n} m_i=g-1
   \label{}
.\end{equation}
Let us write the generating function with the parameter $\hbar $ enumerating contributions of different genera:
\begin{equation}
   F_{\text{BGW}}^\hbar \left( t \right) = \hbar ^2
   \left< \exp \mleft( \sum_{m =0}^{\infty}  (2m+1)  !!
   \hbar ^{2m}t_{2m+1} \tau_{m}\mright)  \right> ^\Theta=
   \sum_{g=0}^{\infty} \hbar ^{2g} F_{\text{BGW}}^g \left( t \right) 
   \label{eq:fbgw}
.\end{equation}
From \eqref{eq:fbgw} one can easily see that the genus expansion is obtained by rescaling of times:
\begin{equation}
   t_k \to t_k \hbar ^{k-1}
   \label{}
.\end{equation}

One can show from definition \eqref{eq:Ql} that $Q$-Schur polynomials in variables $t_k \hbar ^k$ are homogeneous in $\hbar $ and moreover $Q_\lambda\mleft( t_k \hbar ^k \mright) =\hbar ^{|\lambda|}Q_\lambda (t)$, therefore
\begin{equation}
Q_\lambda \mleft(\frac{ t_k\hbar ^{k-1}}{2}\mright) 
	=
	\hbar ^{|\lambda|} Q_{\lambda} \mleft( \frac{t}{2\hbar }\mright) 
	\label{eq:hom}
.\end{equation}

We know that coefficients
\begin{equation}
   C_\lambda=r_\lambda Q_\lambda(\delta_{k,1})
   \label{}
\end{equation}
satisfy BKP Pl\"{u}cker relations because BGW $\tau$-function solves $B$KP. Rescaling $C_\lambda$ by $h^{|\lambda|}$ does not change the relations, since  they are homogeneous by the sum $|\lambda_1|+|\lambda_2|= \mathrm{const}$. As a result, 
\begin{equation}
   C_\lambda^\hbar = \hbar ^{|\lambda|} r_\lambda Q_\lambda(\delta_{k,1})
\end{equation}
satisfy the classical Pl\"{u}cker relations and, hence, expanded
$\tau$-function of spin Hurwitz numbers is a solution of $\hbar $-$B$KP.

By the same argument as in \eqref{eq:hom} we achieve that
\begin{equation}
   \frac{Q_\lambda(\delta_{k,1})}{\hbar ^{|\lambda|}}=
   Q_\lambda\left( \frac{\delta_{k,1}}{\hbar } \right) 
   \label{}
.\end{equation}



Now we can write the $\hbar$-deformed $\tau$-function as
\begin{equation}
   \tau_\text{BGW}^\hbar = \sum_{\lambda \in \mathrm{SP}}^{} r_\lambda \hbar ^{2|\lambda|} Q_\lambda\left( \frac{\delta_{k,1}}{\hbar } \right) Q_\lambda\left( \frac{t}{2\hbar } \right) 
.\end{equation} 
From \eqref{eq:rl} and \eqref{eq:rn} we have
\begin{equation}
   r_\lambda = \prod_{w \in \lambda}^{} \frac{\left( 2c(w)-1 \right) ^2}{16} 
.\end{equation} 
Let us define
\begin{equation}
r_\lambda^\hbar = \hbar ^{2|\lambda|}r_\lambda=
\prod_{w \in \lambda}^{} \frac{\hbar ^2 \left( 2c(w)-1 \right) ^2}{16} 
.\end{equation} 
Now it seems reasonable to introduce
\begin{equation}
   r^\hbar (n)= \frac{\hbar ^2 (2n-1)^2}{16}= \left( \frac{\hbar}{2}\left( n-\frac{1}{2} \right)  \right) ^2
,\end{equation} 
and we are ready to see that rule
\begin{equation}
   n\to \frac{1}{2}+ \hbar \left( n-\frac{1}{2} \right) 
\end{equation} 
allows us to get $r^\hbar (n)$ from $r(n)$ defined in  \eqref{eq:rn}.
That is, we see that $\hbar$-deformation of the BGW model is a solution of $\hbar$-$B$KP hierarchy and the prescription to the deformation is given by
\begin{equation}
   \boxed{
   \left\{
   \begin{aligned}
   &t_k \to \frac{t_k}{\hbar },\\
   &\beta_k \to  \frac{\beta_k}{\hbar },\\
   &r(n) \to  r\left( \frac{1}{2} +\hbar \left( n-\frac{1}{2} \right)  \right) .
   \end{aligned}
   \right.}
   \label{eq:rule}
\end{equation}

In fermionic formalism $\hbar $-BGW $\tau$-function looks like
 \begin{equation}
   \tau^\hbar_{\text{BGW}}(t)=
   \braket{0 | \exp H\mleft( t \mright)\exp \mleft( 
   - \frac{1}{\pi i}\oint \frac{dz}{z}  \phi(-z) 
\frac{\hbar ^2(2D-1)^2}{16z}\phi(z) \mright)  | 0}
   \label{}
.\end{equation}

\subsection{Kontsevich model}\label{ss:kont}
Next, we discuss the Kontsevich model
and its genus expansion. Similarly to the BGW model, the Kontsevich phase of the model was considered in \cite{APSZ20} in the context of the KP hierarchy. This phase does not fit into the family of hypergeometric KP solutions, hence, its deformation in terms of KP hierarchy seems odd. Natural language for the description of the Kontsevich model is the language of hypergeometric $B$KP $\tau$-functions. Here we consider the deformed Konsevich model in these terms and show that it is a solution of the $\hbar $-$B$KP hierarchy. The recipe for the deformation in the context of hypergeometric $B$KP $\tau$-functions is the same as for the BGW model.

\subsubsection{Classical Kontsevich model}
Let us consider the bosonic representation of the Kontsevich model. $\tau$-function of the Kontsevich model \cite{Kon92} in Kontsevich phase \cite{MMS96} is defined by  the matrix integral
\begin{equation}
   Z_{\mathrm{K}}=  \frac{1}{\mathcal{Z}}\int DX \exp \mleft( - \frac{\tr X^3}{3!} - \frac{\tr\Lambda X^{2}}{2}  \mright),\qquad
   \mathcal{Z} =\int DX \exp \mleft( - \frac{\tr \Lambda X^2}{2}  \mright) 
,\end{equation}
where integration is taken over hermitian matrices $X$. From the point of view of $B$KP, it depends only on odd times $t_k$, which are just powers of $\Lambda$
\begin{equation}
t_k = \frac{1 }{k} \tr \Lambda^{-k}
.\end{equation}

According to \cite{MM21,LY22}, the character expansion
of the Kontsevich $\tau$-function in the basis of the Schur $Q$-functions is
\begin{equation}
   \tau_\mathrm{K} (t)= \sum_{\lambda \in \mathrm{SP}}^{} \mleft( \frac{1}{16} \mright) ^{|\lambda|/3}\frac{Q_\lambda(t) Q_\lambda(\delta_{k,1}) Q_{2\lambda}\mleft( \delta_{k,3} /3 \mright) }{Q_{2\lambda }\mleft( \delta_{k,1} \mright)}
   \label{eq:kh}
.\end{equation} 
Therefore, we have
\begin{equation}
   \tau_\mathrm{K} \mleft(t\mright)= \sum_{\lambda \in  \mathrm{SP}}^{}r_\lambda Q_\lambda \mleft( \frac{\delta_{k,3}}{3} \mright) Q_\lambda\mleft(t\mright),\qquad r_\lambda= \mleft( \frac{1}{16} \mright) ^{|\lambda| /3}\frac{Q_\lambda\mleft( \delta_{k,1} \mright) Q_{2\lambda}\mleft( \delta_{k,3}  /3 \mright) }{Q_{2\lambda}\mleft( \delta_{k,1} \mright) Q_{\lambda}\mleft( \delta_{k,3} /3 \mright) }.
\end{equation} 
For our purposes, we have to find $r(n)$ function for given
$r_\lambda$. Now let us derive it. The answer is formulated in \eqref{eq:rnk}. 
From \cite{Ale23} we know that
\begin{equation}
   \frac{Q_\lambda \mleft( \delta_{k,1} \mright)}{Q_{2\lambda}(\delta_{k,1})}= \prod_{w \in \lambda}^{} 
   \mleft(2c(w)-1\mright)  
   \label{}
.\end{equation}
Therefore, we have
\begin{equation}
   \mleft( \frac{1}{16} \mright) ^{|\lambda| /3}
   \frac{Q_\lambda\mleft( \delta_{k,1} \mright) }{Q_{2\lambda}\mleft( \delta_{k,1} \mright) }=
   \prod_{w \in \lambda}^{}  \mleft( \frac{1}{16} \mright) ^{1 /3}\mleft(2c(w)-1\mright) 
   \label{}
.\end{equation}
Now it is convenient for us to express any partition $\lambda$ in the following form
\begin{equation}
  \lambda= \mleft( 3k_1,\ldots,3k_p,3m_1+1,\ldots,3m_q+1,
  3n_1+2,\ldots,3n_r+2\mright)  .
   \label{}
\end{equation}
The remain part of $r_\lambda$ is given by the ratio of Q-Schur polynomials at the point $\delta_{k,3}$ (see \cite{LY22} or \cite{MMNO21}):
\begin{multline}
   \frac{Q_{2\lambda}\mleft( \delta_{k,3} /3 \mright) }{Q_{\lambda}\mleft( \delta_{k,3} /3 \mright) }=
   \frac{(-1)^r (1 /3)^{|\lambda| /3}}{\prod_{i=1}^{p} \mleft( 2k_i-1 \mright) !! \prod_{i=1}^{q} (2m_i-1)!! \prod_{i=1}^{r} (2n_i+1)!!   }\\=
   (-1)^r \mleft( \frac{1}{3} \mright) ^{|\lambda| /3}
   \mleft(\prod_{i=1}^{p} \prod_{j=1}^{k_i} (2j-1)    \mright)^{-1} \mleft(\prod_{i=1}^{q} \prod_{j=1}^{m_i} (2j-1)    \mright)^{-1} \\\times\mleft(\prod_{i=1}^{r} \prod_{j+1=1}^{n_i} (2(j+1)-1)    \mright)^{-1} 
   \label{}
.\end{multline}
Our goal is to rewrite this expression in terms of the product over the boxes of the Young diagram. It can be thought of in the following way. We can by default fill $\lambda$ with 1 in each box and associate particular multipliers in the last equation with particular boxes. Each of $k_i$, $m_i$, $n_i$, corresponds to the particular row of $\lambda$ and it is natural for beginning to associate products up to $k_i$ or   $m_i$ or  $n_i$ with the corresponding row. In turn, for each picked row we can associate the particular multiplier with the particular box. We propose to fill three types of rows (associated with $k_i$ or $m_i$ or  $n_i$) in the following way

\begin{equation}
   \begin{gathered}
      \underbrace{\begin{tabular}{|l|l|l|l|l|l|l|l|l|l|}
	    \cline{1-5} \cline{7-10}
	    $-1$& $-(2\cdot 1-1)^{-1}$ & $1$ &$-1$ & $-(2\cdot 2-1)^{-1}$ & $\cdots$ & $1$& $-1$& $-(2k_i-1)^{-1}$& $1$  \\ 	 \cline{1-5} \cline{7-10}
      \end{tabular}}_{3k_i}\\
      \underbrace{\begin{tabular}{|l|l|l|l|l|l|l|l|l|l|l|}
	    \cline{1-5} \cline{7-10}
	    $-1$ &$-(2\cdot 1-1)^{-1}$& $1$& $-1$  & $-(2\cdot 2-1)^{-1}$& $\cdots$ &$-1$&    $-(2m_i-1)^{-1}$&$1$& $-1$ \\ 
	    \cline{1-5} \cline{7-10}
      \end{tabular}}_{3m_i+1}\\
      \underbrace{\begin{tabular}{|l|l|l|l|l|l|l|l|l|l|l|}
	    \cline{1-5} \cline{7-9}
	    $-1$ & $-(2\cdot 1-1)^{-1}$& $1$&  $-1$ & $-(2\cdot 2-1)^{-1}$& $\cdots$  &$1$& $-1$& $-(2(n_i+1)-1)^{-1}$ \\ \cline{1-5} \cline{7-9} 
      \end{tabular}}_{3n_i+2}\\
   \end{gathered}
   \label{eq:rows}
\end{equation}
All minuses are here to get the overall factor  $(-1)^r$. 
Because in fact the number of all minuses in the partition $\lambda$ is
\begin{equation}
   \sum_{i=1}^{p} 2k_i+ \sum_{i=1}^{q} (2m_i+1)+
   \sum_{i=1}^{r} (2n_i+2)= 2\mleft( \ldots \mright) +q+2r \xlongequal[]{q=r}
   2\mleft( \ldots \mright) +3r
   \label{}
.\end{equation}
The last equality holds since \cite{LY22} $Q$-Schur polynomials
$Q_\lambda (\delta_{k,3})$ are nonzero only for diagrams with $q=r$. The first example of strict partition with $3n$ boxes, but 
which has $q\neq r$ is $\mleft[ 7,4,1 \mright] $ 
\begin{center}
\ydiagram{7,4,1}
\end{center}
And for such diagram Schur Q-function at the point $\delta_{k,3}$ is zero. After we calculate the product of all minuses, we get
\begin{equation}
   (-1)^{2\mleft( \ldots \mright) +3r}=(-1)^r 
   \label{}
.\end{equation}
After all
\begin{multline}
   \frac{Q_{2\lambda}\mleft( \delta_{k,3} /3 \mright) }{Q_{\lambda}\mleft( \delta_{k,3} /3 \mright) }=  
   \prod_{i=1}^{\ell(\lambda)} \prod_{j=1}^{\lambda_i} \mleft( \frac{1}{3} \mright) ^{1  /3}
   \begin{cases}
      1, & j \bmod 3=0,\\
      -1, & j \bmod 3=1,\\
      -\mleft( 2 \frac{j+1}{3}-1 \mright)^{-1} , & j \bmod 3=2,\\
   \end{cases}\\=\prod_{i=1}^{\ell(\lambda)} \prod_{j=1}^{\lambda_i} \mleft( \frac{1}{3} \mright) ^{1  /3}
   \begin{cases}
      1, & j \bmod 3=0,\\
      -1, & j \bmod 3=1,\\
      \frac{3}{1-2j} , & j \bmod 3=2.\\
   \end{cases}
   \label{}
\end{multline}
Note that rows of $\lambda$ \eqref{eq:rows}
are filled in
agreement with the function on the r.h.s. excluding power of $1 /3$,
which is trivial to insert.

$B$KP-content is given by \eqref{eq:cont}, and now we can restore
\begin{equation}
   r_\lambda= \prod_{w \in \lambda}^{} 
   \mleft( \frac{1}{16} \mright) ^{1 /3}\mleft(2c(w)-1\mright) 
   \mleft(\frac{1}{3}\mright)^{1 /3}\begin{cases}
      1, & c(w) \bmod 3=0,\\
      -1, & c(w) \bmod 3=1,\\
      \frac{3}{1-2c(w)} , & c(w) \bmod 3=2.\\
   \end{cases}
   \label{}
\end{equation}

After looking once again on proposed way of filling partition $\lambda$ \eqref{eq:rows} (remember that $q=r$) the simplified function $r(n)$ looks like
\begin{equation}
   r(n)=\mleft( \frac{1}{16} \mright) ^{1/3}
   \begin{cases}
      2n-1, & n \bmod 3=0,\\
      1-2n, & n \bmod 3=1,\\
      -1 , & n \bmod 3=2.\\
   \end{cases}
   \label{eq:rnk}
\end{equation}
In the fermionic formalism the Kontsevich $\tau$-function has the following form
 \begin{equation}
   \tau_{\text{K}}\mleft( \frac{t}{2} \mright) =
   \braket{0 | \exp H\mleft( t \mright)\exp \mleft( 
   - \frac{1}{3\pi i}\oint \frac{dz}{z}  \phi(-z) 
\mleft( \frac{1}{z}r(D) \mright) ^3\phi(z) \mright)  | 0}
   \label{}
.\end{equation}

\subsubsection{Kontsevich model as a solution of \texorpdfstring{$\hbar $}{ħ}-\texorpdfstring{$B$}{B}KP}
The Kontsevich model is the generating function for intersection
numbers of Chern classes on compactified moduli
spaces $\overline{\mathcal{M}}_{g,n}$ of complex
curves of genus $g$ with $n$ marked points. Intersection numbers
of Chern classes
\begin{equation}
   \int_{\overline{\mathcal{M}}_{g,n}}\psi_1^{m_1} \psi_{2}^{m_2}
   \cdots \psi_n^{m_n} = \left< \tau_{m_1} \tau_{m_2} \cdots
   \tau_{m_n}\right> 
   \label{}
\end{equation}
are rational numbers, which are not equal to zero only if
\begin{equation}
   \sum_{i=1}^{n} (m_i-1)=3g-3
   \label{}
.\end{equation}
Let us write the generating function with inserted parameter $\hbar $ 
enumerating contributions of different genera \cite{Ale14}:

\begin{equation}
   F_\mathrm{K}^\hbar \left( t \right) = \hbar ^2
   \left< \exp \mleft( \sum_{m =0}^{\infty}  (2m+1)  !!
   \hbar ^{\frac{2(m-1)}{3}}t_{2m+1} \tau_{m}\mright)  \right> =
   \sum_{g=0}^{\infty} \hbar ^{2g} F_\mathrm{K}^g \left( t \right) 
   \label{eq:fk}
.\end{equation}
From \eqref{eq:fk} we can see that the genus expansion is obtained
by rescaling times:
\begin{equation}
t_k \to \hbar ^{\frac{k-3}{3}}t_k
.\end{equation} 
After this rescaling one can obtain
\begin{equation}
   \tau_\mathrm{K}(t)= \sum_{\lambda  \in \mathrm{SP}}^{} r_\lambda \hbar ^{|\lambda| /3} 
   Q_{\lambda}\mleft( \frac{\delta_{k,3}}{3} \mright) Q_\lambda\mleft( \frac{t}{\hbar} \mright) = \sum_{\lambda \in \mathrm{SP}}^{} \hbar ^{2|\lambda| /3}r_\lambda Q_{\lambda} \mleft( \frac{\delta_{k,3}}{3\hbar } \mright) Q_{\lambda }\mleft( \frac{t}{\hbar } \mright) 
   \label{eq:tauk}
.\end{equation}
Now we have $r^\hbar _\lambda= \hbar ^{2|\lambda| /3}r_\lambda$
and using the same technique as in case of BGW model we claim that genus expanded Kontsevich model solves $\hbar $-$B$KP hierarchy.
Next, by the same arguments as in the previous subsection
we get 
\begin{equation}
      r^\hbar  (n)= \mleft( \frac{1}{16} \mright) ^{1 /3} \begin{cases}
	  \hbar (2n-1), & n\bmod{3} =0,\\
      \hbar (1-2n), & n \bmod 3 = 1,\\
      -1,& n \bmod 3 =2.
   \end{cases}
   \label{}
\end{equation}
Note that this $\hbar$ insertion inside $r(n)$ is true only for such $\lambda$ that have $q=r$, and the other Young diagrams do not contribute to the partition function. That is, the recipe for deformation of hypergeometric $B$KP $\tau$-functions \eqref{eq:rule} holds for the Kontsevich model as well.

In the fermionic formalism $\hbar $-deformed Kontsevich $\tau$-function has the form
 \begin{equation}
   \tau^\hbar_{\text{K}}\mleft(\frac{t}{2}\mright)=
   \braket{0 | \exp H\mleft( t \mright)\exp \mleft( 
   - \frac{1}{3\pi i}\oint \frac{dz}{z}  \phi(-z) 
 \mleft( \frac{1}{z} r^\hbar \mleft(  D\mright)  \mright)^3 \phi(z) \mright)  | 0}
   \label{}
.\end{equation}

\section{Spin Hurwitz numbers as \texorpdfstring{$\hbar $}{ħ}-\texorpdfstring{$B$}{B}KP solution and its \texorpdfstring{$\hbar $}{ħ}-KP counterpart}
\label{s:shn}
This section is devoted to the generating functions for Hurwitz numbers with completed cycles,
and their spin counterpart.

Ordinary Hurwitz numbers, counting ramified coverings of a Riemann surface with imposed conditions on the ramifications, were defined by Hurwitz in \cite{Hur91,Hur01}. In the more recent years, Hurwitz numbers again became an object of interest, due to strong ties with the integrable hierarchies \cite{Oko00}, Gromov-Witten theory \cite{OP06}, the intersection theory of the moduli spaces of curves via ELSV type formulae \cite{ELSV01}, topological recursion \cite{EO08,BM08,BEMS11}, and $W$-representation of partition functions \cite{MMN09}. 

In this section, we in particular consider a type of Hurwitz numbers called spin Hurwitz numbers, introduced by Eskin-Okounkov-Pandharipande \cite{EOP08}. The defining feature of these numbers is the presence of a spin structure (or theta characteristic) on the surfaces, and the counting of coverings is weighted by the parity of this theta characteristic. We denote spin Hurwitz numbers with a superscript $\vartheta$, to emphasize the role of the theta characteristic.

It is worth mentioning that the term ``$r$-spin Hurwitz numbers'' is also used (in e.g. \cite{MSS13,SSZ15,BKL21,KLPS19,DKPS19}) for what we call ``ordinary Hurwitz numbers with completed cycles''. We do not use the term $r$-spin here. Their spin counterparts --- ``spin Hurwitz numbers with completed cycles'' were recently actively studied in \cite{Gun16,Lee19,MMN20,MMNO21,GKL21,AS21, MMZ21, MMZ22}.

In this section we find $f(n)$ and $r(n)$ functions for both generating functions for ordinary Hurwitz numbers with completed cycles and their spin counterpart. We formulate the prescription of $\hbar$ insertion into the functions $f(n)$ and $r(n)$. The prescription differs from \eqref{eq:rule}. We explicitly check that the genus expansion of these models in terms of $\hbar$ satisfies the $\hbar $-$B$KP hierarchy.

\subsection{Hurwitz numbers with completed cycles}
\label{ss:hncc}
This part of the section is devoted to a topic that is directly connected with the KP hierarchy and its $\hbar $-deformation. As it was mentioned in the introduction, such a revisit of $\hbar$-KP examples is important in understanding the $\hbar$-$B$KP case. We refer the reader to \cite{APSZ20} for a detailed introduction to the subject. Here we only fix some notation and immediately after that we switch to the discussion of Hurwitz numbers.

\subsubsection{Classical ordinary Hurwitz numbers with completed cycles}
Schur polynomials $S_\lambda(t)$ are defined with the help of generating function and determinant formula:
 \begin{equation}
   S_\lambda(t) = \det _{i,j} S_{\lambda_i-i+j}(t) ,\qquad \sum_{k}^{} S_k (t) z^k= \exp \mleft( \sum_{k\ge 1}^{} t_k z^k \mright)
   \label{}
.\end{equation}
Partition functions $Z=Z\mleft( t, u \mright)$ for ordinary Hurwitz numbers with completed cycles \cite{MMN09} is given by
\begin{equation}
   Z =
   \sum_{\lambda}^{} \exp \mleft({\sum_{r>0}^{} u_r p_r (\lambda)}\mright) S_\lambda \mleft(\delta_{k,1} \mright)S_\lambda\mleft(t\mright) ,\qquad 
   \label{eq:hurw_part}
\end{equation}
where the symmetric sums (or Casimirs) are defined as
\begin{equation}
   p_r(\lambda)= \sum_{i=1}^{\ell(\lambda)} \mleft(\lambda_i-i+\frac{1}{2}\mright)^r-\mleft( -i+\frac{1}{2} \mright) ^r
   \label{}
.\end{equation}
Now let us find the $f(n)$ function for
 \begin{equation}
    f_\lambda= \exp \mleft( \sum_{r>0}^{} u_r p_r(\lambda) \mright) 
   \label{}
.\end{equation}
To do this we can start with the natural ansatz 
\begin{equation}
   f(n)=  \exp \mleft( \sum_{r>0}^{} u_r p_r(n) \mright)  
   \label{}
.\end{equation}
Identity (KP content differs from $B$KP analogue \eqref{eq:cont} and equals  $c(w)=j-i$)
\begin{equation}
   f_\lambda= \prod_{i=1}^{\ell(\lambda)} \prod_{j=1}^{\lambda_i} 
   f(j-i)
   \label{}
\end{equation}
must hold for all $\lambda$ and $u_k$, therefore
\begin{equation}
   \mleft( \lambda_i-i+\frac{1}{2} \mright) ^r-\mleft( \frac{1}{2}-i \mright) ^r=\sum_{j=1}^{\lambda_i} p_r(j-i)
   \label{}
.\end{equation}
And it is easy to check that function
\begin{equation}\label{hurw_casimir_fn}
   p_r(n)= \mleft(n+\frac{1}{2}\mright)^r-\mleft( n-\frac{1}{2} \mright) ^r
,\end{equation}
solves this functional equation.
One can see that the set of $p_r(n)$,  $r>0$ is the basis
in space of infinitely differentiable functions 
and after all $f(n)$ for Hurwitz numbers with completed
cycles is an arbitrary infinitely differentiable function.
Therefore, every partition function $Z$ for Hurwitz numbers
with completed cycles is hypergeometric KP $\tau$-function
and vice versa.

\subsubsection{Ordinary Hurwitz numbers with completed cycles as a solution of \texorpdfstring{$\hbar $}{ħ}-KP}

Hurwitz numbers with $(r+1)$-completed cycles defined by partition $\mu$ and genus $g$ of the covering surface again can be expressed as the integrals over the moduli space of curves \cite{SSZ15,DKPS19}
\begin{equation}
   h_{g,\mu}^r=b!r^{2g-2+n+b} \mleft( \prod_{i=1}^{n} \frac{\mleft( \frac{\mu_i}{r} \mright) ^{\mleft[ \mu_i \mright] }}{\mleft[ \mu_i \mright] !}  \mright) \int_{\overline{\mathcal{M}}_{g,n}}
      \frac{C\mleft( r,1; \mleft<\mu \mright>  \mright) }{
   \prod_{i=1}^{n} \mleft( 1- \frac{\mu_i}{r}\psi_i \mright)}
   \label{}
,\end{equation}
where
\begin{equation}
   n=\ell(\mu),\qquad b= \frac{2g-2+\ell(\mu)+\left| \mu \right| }{r},\qquad \mu_i =r \mleft[ \mu_i \mright] +\langle \mu_i \rangle 
   \label{eq:defn}
,\end{equation}
and $C$ is the Chiodo class \cite{Chi06}.
The generating function for Hurwitz numbers with $(r+1)$-completed cycles and $\hbar$ insertion is then given by
\begin{equation}
   F_r^\hbar = \sum_{\mu}^{} \sum_{b}^{}\hbar ^{2g}  h^{r}_{g,\mu} p_\mu \frac{u^b}{b!}
   \label{}
,\end{equation}
where $p_{\mu} = \prod_{k=1}^{\ell(\mu)} p_{\mu_k}$ and $p$-variables (commonly used for Hurwitz generating functions) are related to the $t$-variables (commonly used for integrable hierarchies) via simple rescaling $p_{k} = kt_{k}$. The corresponding partition function with inserted $\hbar$ is then
\begin{equation}
   Z^\hbar _r(t,u)=\sum_{\mu}^{} \sum_{b}^{} \hbar ^{br - \ell(\mu) - |\mu|} h^{\bullet,r}_{g,\mu}  p_\mu \frac{u^b}{b!}
   \label{}
.\end{equation}
Note that the partition function is the generating function for the disconned Hurwitz numbers $h^{\bullet,r}_{g,\mu}$, where $g$ is defined for disconnected surfaces via \eqref{eq:defn}. Now let us rearrange the terms in the $\hbar$-deformed partition function to bring it to the form \eqref{eq:hurw_part}. One of the most important ingredients in our calculations is the Frobenius formula \cite{Fro96}
 \begin{equation}
    S_\lambda(p)= \sum_{\mu\vdash |\lambda| }^{}  \frac{\chi^{\lambda}_\mu}{z_{\mu}} p_\mu,\qquad
    z_{\mu} =\prod_{k}^{} k^{m_{k}} m_k! 
    \label{}
,\end{equation}
where $m_k$ is the number of lines of length  $k$ and $\chi_\mu^\lambda$ are the characters of the symmetric group.
Combinatorial definition of disconnected Hurwitz numbers with $(r+1)$-completed cycles \cite{OP06,SSZ12}
\begin{equation}
   h_{g,\mu}^{\bullet,r}=   \sum_{\lambda\vdash|\mu|}^{} \frac{\operatorname{dim} \lambda}{\left| \mu \right| !} \frac{\chi_{\mu}^\lambda}{z_\mu} \mleft( \frac{p_{r+1}(\lambda)}{(r+1)!} \mright) ^b,\qquad \operatorname{dim} \lambda = \chi^\lambda_{\mleft[ 1^{\left| \lambda \right| } \mright] },
   \label{}
\end{equation}
can be used for the simplification of the partition function
\begin{multline}
   Z^\hbar _r=\sum_{\mu}^{} \sum_{b}^{}   \sum_{\lambda\vdash|\mu|}^{} \hbar ^{br - \ell(\mu) - |\mu|} \frac{\operatorname{dim}\lambda}{|\mu|!}  \frac{\chi_{\mu}^\lambda }{z_\mu}\mleft( \frac{p_{r+1}(\lambda)}{(r+1)!} \mright) ^b  p_\mu \frac{u^b}{b!}\\=
\sum_{\lambda}^{} \exp \mleft( \hbar ^ru\frac{p_{r+1}(\lambda)}{(r+1)!} \mright) S_\lambda \mleft(\frac{\delta_{k,1}}{\hbar }\mright) S_\lambda\mleft( \frac{p}{\hbar} \mright)  
   \label{}
.\end{multline}
Next, similarly to the non-deformed case, in order to obtain the KP $\tau$-function of the form \eqref{eq:hurw_part} one has to go back to the $t$-variables
\begin{equation}
 Z^\hbar_r =\sum_{\lambda}^{} \exp \mleft( \hbar^{r} u \frac{p_{r+1}(\lambda)}{(r+1)!} \mright) S_\lambda \mleft(\frac{\delta_{k,1}}{\hbar }\mright)S_\lambda\mleft( \frac{t}{\hbar } \mright)   
   \label{}
.\end{equation}
For $(r+1)$-completed cycles KP Pl\"{u}cker coefficients are
\begin{equation}
   C_\lambda^\hbar =(f_\lambda)^{\hbar ^{r}}S_\lambda \mleft( \frac{\delta_{k,1}}{\hbar } \mright) 
   \label{}
.\end{equation}
So every term in Pl\"{u}cker relations will be of the form
\begin{equation}
   C_{\lambda_1}^\hbar C_{\lambda_2}^\hbar =\mleft( f_{\lambda_1} f_{\lambda_2} \mright) ^{\hbar ^{r}}
   S_{\lambda_1}\mleft( \frac{\delta_{k,1}}{\hbar } \mright) 
   S_{\lambda_2}\mleft( \frac{\delta_{k,1}}{\hbar } \mright) 
   \label{}
.\end{equation}
One can show that for the KP Pl\"{u}cker relations the first factor is the same for all terms. Therefore, since Schur polynomials satisfy Pl\"{u}cker relations, $\tau$-function of Hurwitz numbers with  $\mleft( r+1 \mright) $-completed cycles solves $\hbar $-KP. This fact easily generalizes to the arbitrary linear combinations of completed cycles.

To summarize, $Z^\hbar $ solves  $\hbar $-KP and, in particular, for each $p_r(n)$:
 \begin{equation}
    p_{r}^\hbar (n)=\hbar ^{r-1}\mleft( \mleft( n+\frac{1}{2} \mright) ^{r}- \mleft( n-\frac{1}{2} \mright) ^{r}
 \mright)    \label{}
.\end{equation}
We observe the following deformation rule for the ordinary Hurwitz numbers with completed cycles:
\begin{equation}
   \boxed{p_r(n)\to \hbar ^{r-1} p_r(n),\qquad \beta_k \to  \frac{\beta_k}{\hbar },\qquad t_k \to \frac{t_k}{\hbar }}
   \label{hndef}
\end{equation}
and it differs from \eqref{KPhyp}. Let us mention that for the simple Hurwitz numbers, considered in \cite{APSZ20}, $u_r=\frac{u}{2}\delta_{r,2}$ and the deformation rule \eqref{hndef} simplifies to
\begin{equation}
   f^\hbar (n) = e^{\hbar u n}=f(\hbar n)
   \label{}
,\end{equation}
but $f^\hbar (n) \neq f(\hbar n)$ for any other choices of $u$
with higher $r$. The reason for such a difference in $\hbar$-deformation prescriptions is in the form of function $f(n)$. From \cite{BDKS20} we know that $f(n)$ is tightly connected with the spectral curve data. Since the spectral curve (in fact, the whole procedure of topological recursion) contains the information about correlators corresponding to different genera, the $\hbar$-deformation procedure can be encoded into this data as well. However, from this point of view, the form of the function $f(n)$ (either it is rational or exponential) appears to be extremely important and it changes the way of $\hbar$-deformation.

Finally, let us write the $\hbar$-deformed Hurwitz $\tau$-function in the fermionic formalism:
\begin{equation}
   p_{r}^\hbar (n)=\hbar ^{r-1}\mleft(\mleft( n+\frac{1}{2} \mright) ^{r}-\mleft( n-\frac{1}{2} \mright) ^{r}\mright) = \sum_{m=0}^{\mleft\lfloor r /2 \mright\rfloor} \binom{ r}{2m+1}\mleft( \frac{\hbar }{2} \mright) ^{2m} (\hbar n)^{r-2m-1}   \label{}
.\end{equation}
In the fermionic formalism $p_r^\hbar (n)$ corresponds 
to the operator  $p_r^\hbar (D)$
 \begin{equation}
    p_r^\hbar (D)=\sum_{m=0}^{\mleft\lfloor r /2 \mright\rfloor} \binom{ r}{2m+1}\mleft( \frac{\hbar }{2} \mright) ^{2m} (\hbar D)^{r-2m-1}
   \label{}
\end{equation}
and we can write the generating function of $(r+1)$-completed Hurwitz numbers as
\begin{equation}
   \tau^\hbar (t) = \braket{0 | \exp \mleft( \frac{H(t)}{\hbar }\mright)
   \exp \mleft( \frac{A_r^\hbar}{\hbar }  \mright)   | 0} 
   \label{}
,\end{equation}
where
\begin{equation}
   A^\hbar_r =\oint \frac{dz}{2\pi i}\normord{\mleft( \frac{1}{z}\exp\mleft( p_r^\hbar (D)\mright) \psi(z) \mright) \psi^*(z)}
   \label{}
.\end{equation}
Such $\hbar$ insertion implies that $\tau$-function has good quasi-calssical limit according to \cite{TT95}.

\subsection{Spin Hurwitz numbers with completed cycles}
\label{ss:shncc}
Finally, we consider the last family of $B$KP  $\tau$-functions, in full analogy with the KP case.
\subsubsection{Classical spin Hurwitz numbers with completed cycles}
Partition function $\mathcal{Z}=\mathcal{Z}\mleft( t,  u \mright)$ for spin Hurwitz numbers with completed
cycles \cite{MMN20} is given by
\begin{equation}
   \mathcal{Z} =
   \sum_{\lambda \in \mathrm{SP}}^{} \exp \mleft({\sum_{r \in \mathbb{Z}_{\text{odd}}^+}^{} u_r \mathbf{p}_r (\lambda)}\mright)Q_\lambda \mleft( \frac{\delta_{k,1}}{2} \mright)Q_\lambda\mleft(\frac{t}{2}\mright)  ,
   \label{}
\end{equation}
where the $B$KP symmetric sum is given by	
\begin{equation}
   \mathbf{p}_r(\lambda)= \sum_{i=1}^{\ell(\lambda)} \lambda_i^r 
   \label{}
.\end{equation}
As for the non-spin case, we are going to find the $r(n)$ function for
 \begin{equation}
    r_\lambda= \exp \mleft( \sum_{r \in \mathbb{Z}_\text{odd}^+}^{} u_r \mathbf{p}_r(\lambda) \mright) 
   \label{}
.\end{equation}
Let us make an analogous  ansatz 
\begin{equation}
   r(n)=  \exp \mleft( \sum_{r \in \mathbb{Z}_{\text{odd}}^+}^{} u_r \mathbf{p}_r(n) \mright)  
   \label{}
.\end{equation}
The identity
\begin{equation}
   r_\lambda= \prod_{i=1}^{\ell(\lambda)} \prod_{j=1}^{\lambda_i} 
   r(j)
   \label{}
\end{equation}
must hold for all $\lambda$ and $u_r$, therefore
\begin{equation}
   \lambda_i^r=\sum_{j=1}^{\lambda_i} \mathbf{p}_r(j)
   \label{}
.\end{equation}
And we obtain
\begin{equation}
   \mathbf{p}_r(n)= n^r-\mleft( n-1 \mright) ^r
   \label{}
.\end{equation}
Note that the function $ \mathbf{p}_r(n) $ differs from the function $p_{r}(n)$, given by \eqref{hurw_casimir_fn}, by the simple shift $n \to n - \frac{1}{2}$. Now, the set of $\mathbf{p}_r(n)$,  $r \in \mathbb{Z}_{\text{odd}}^+$ is the basis
in space of infinitely differentiable functions which are symmetric with respect to 1/2
and  $r(n)$ is an arbitrary function in this space.
Therefore, every partition function $\mathcal{Z}$ for spin Hurwitz numbers
with completed cycles is hypergeometric $\tau$-function of $B$KP
and vice versa.

\subsubsection{Spin Hurwitz numbers with completed cycles as a solution of \texorpdfstring{$\hbar $}{ħ}-\texorpdfstring{$B$}{B}KP}
Spin Hurwitz numbers with $(r+1)$-completed cycles also
have a representation in terms of integrals over the moduli spaces
of curves \cite{GKL21,AS21}
\begin{equation}
   h_{g,\mu}^{r,\vartheta}= b! r^{2g-2+n+b}\mleft( \prod_{i=1}^{n} \frac{\mleft( \frac{\mu_i}{r} \mright) ^{\mleft[ \mu_i \mright] }}{\mleft[ \mu_i \mright] !}  \mright) \int_{\overline{\mathcal{M}}_{g,n}} \frac{C^\vartheta \mleft( r,1; \langle \mu \rangle  \mright) }{\prod_{i=1}^{n} \mleft( 1- \frac{\mu_i}{r}\psi_i \mright)  }
   \label{}
,\end{equation}
where we use definitions from \eqref{eq:defn}
and $C^\vartheta$ is the Chiodo class
twisted by the 2-spin Witten class.
Free energy and partition function for spin Hurwitz numbers with $(r+1)$-completed cycles are defined in full analogy with the non-spin case. The difference is that one should replace $h_{g,\mu}^{r} \to h_{g,\mu}^{r,\vartheta}$, $\mu$ should run over the set of odd partitions (OP) and $\lambda$ over the set of strict partitions (SP).
Here we need to introduce the characters of the Sergeev group $\zeta_\mu^\lambda$ in a similar to the Frobenius formula way \cite{Ser85}
\begin{equation}
    Q_\lambda(p)= 2^{-\frac{1}{2}\delta(\lambda)} \sum_{\mu \in  \mathrm{OP}}^{} \frac{\zeta_\mu^\lambda}{z_{\mu}} p_\mu
    \label{}
,\end{equation}
where $\delta(\lambda)$ is equal to 0 or 1 if $\ell(\lambda)$ is even or odd respectively. Note that we have the factor of $2^{-\frac{1}{2}\ell(\lambda)}$ in the definiton of $Q$-Schur polynomials \eqref{eq:Ql}. Here, again, $p_k = k t_k$.

Combinatorial definition for disconnected spin Hurwitz numbers with $(r+1)$-completed cycles \cite{Gun16,Lee19}
\begin{equation}
   h_{g,\mu}^{\bullet,r,\vartheta}=   \sum_{\lambda \in \mathrm{SP}}^{} \frac{\operatorname{dim} \lambda}{2^{\delta(\lambda) + \ell(\lambda) + |\lambda|} \left| \lambda \right| !} \frac{\zeta_{\mu}^\lambda}{z_\mu} \mleft( \frac{\mathbf{p}_{r+1}(\lambda)}{(r+1)!} \mright) ^b,\qquad \operatorname{dim} \lambda = \zeta^\lambda_{\mleft[ 1^{\left| \lambda \right| } \mright] },
   \label{}
\end{equation}
allows us to rewrite the partition function in the $t$-variables (and standard for $B$KP rescaling $t_{k} \to t_{k}/2$) as
\begin{equation}
   \mathcal{Z}^\hbar_r =\sum_{\lambda \in \mathrm{SP}}^{} \exp \mleft( \hbar ^{r} u \frac{ \mathbf{p}_{r+1}(\lambda) }{(r+1)!} \mright) Q_\lambda \mleft(\frac{\delta_{k,1}}{2\hbar }\mright)Q_\lambda\mleft( \frac{t}{2\hbar } \mright)   
   \label{}
.\end{equation}
Thus, $B$KP Pl\"{u}cker coefficients for $(r+1)$-completed cycles are
\begin{equation}
    C^{\hbar}_{\lambda} = (r_{\lambda})^{\hbar^r} Q_\lambda \mleft(\frac{\delta_{k,1}}{2\hbar }\mright)
\end{equation}
and from the general form of $B$KP Pl\"{u}cker relations \eqref{eq:plucker_general} it is easy to see that every term is of the form
\begin{equation}
    C^{\hbar}_{\lambda_1} C^{\hbar}_{\lambda_2} = (r_{\lambda_1} r_{\lambda_2})^{\hbar^r} Q_{\lambda_1} \mleft(\frac{\delta_{k,1}}{2\hbar }\mright) Q_{\lambda_2} \mleft(\frac{\delta_{k,1}}{2\hbar }\mright)
\end{equation}
The first factor is the same for all the terms. Using that $Q$-Schur polinomials solve the $B$KP Pl\"{u}cker relations, we obtain that $\tau$-function of spin Hurwitz numbers with $(r+1)$-completed cycles solves $\hbar$-$B$KP. The generalization to arbitrary linear combination of completed cycles is straightforward.

To summarize, $\mathcal{Z}^{\hbar}$ solves $\hbar$-$B$KP and, in particular, for each $\mathbf{p}_r(n)$:
 \begin{equation}
    \mathbf{p}_{r}^\hbar (n)=\hbar ^{r-1} \mleft(n  ^{r}-  \mleft( n-1 \mright) ^{r}\mright)
   \label{}
.\end{equation}
We observe the following deformation rule for the spin Hurwitz numbers with completed cycles:
\begin{equation}
   \boxed{\mathbf{p}_r(n)\to \hbar ^{r-1} \mathbf{p}_r(n),\qquad \beta_k \to  \frac{\beta_k}{\hbar },\qquad t_k \to \frac{t_k}{\hbar }}
   \label{shndef}
\end{equation}
The deformation rule \eqref{shndef} in terms of function $\mathbf{p}_{r}(n)$ is identical to the one of the ordinary Hurwitz numbers \eqref{hndef}. However, similarly to the KP case, it differs from BGW- and Kontsevich-type deformation prescription \eqref{eq:rule} because
\begin{equation}
   r^\hbar (n) \neq r\mleft(\hbar \mleft(n-\frac{1}{2}\mright)+\frac{1}{2}\mright)
\end{equation}
for any choice of $u$ but the trivial case $u_r=u\delta_{r,1}$. So, for the spin Hurwitz numbers we observe the same discrepancy in $\hbar$-deformation as it was for the ordinary Hurwitz numbers. The reasonings in this case should be similar: either rational or exponential form of the function $r(n)$ implies different $\hbar$ insertion into the spectral curve data, which is in agreement with \cite{AS21}.

Finally, let us write the $\hbar$-deformed spin Hurwitz $\tau$-function in the fermionic formalism:
\begin{equation}
   \mathbf{p}_{r}^\hbar (D)=
    \sum_{m=0}^{\frac{r-1}{2}} \binom{ r}{2m+1}\mleft( \frac{\hbar }{2} \mright) ^{2m} \mleft(\hbar \mleft(D-\frac{1}{2}\mright)\mright)^{r-2m-1}   \label{}
.\end{equation}
So, for spin Hurwitz numbers with $(r+1)$-completed cycles we have
\begin{equation}
   \tau^\hbar_{}(t)=
   \braket{0 | \exp \mleft( \frac{H\mleft( t \mright)}{\hbar } \mright) \exp \mleft( 
    \frac{\mathcal{A}_r^\hbar }{\hbar }\mright)  | 0}
   \label{}
,\end{equation}
where
\begin{equation}
   \mathcal{A}^\hbar _r=- \frac{1}{4\pi i}\oint \frac{dz}{z^2}  \phi(-z) 
   \exp \mleft( \mathbf{p}_r^\hbar (D) \mright) \phi(z)
   \label{}
.\end{equation}

\section{Conclusion}
To summarize, the main results of the paper are:
\begin{itemize}
\item We have considered the genus expansion of several $B$KP $\tau$-functions governed by parameter $\hbar$. Among the examples were Kontsevich and BGW models and generating functions for spin Hurwitz numbers with completed cycles. We have shown, that all these $\tau$-functions with inserted parameter $\hbar$ are solutions of $\hbar $-$B$KP with the correct quasi-classical behavior.

\item We have considered all the mentioned examples as members of the hypergeometric $B$KP family, each member of which is parametrized by a single function $r(n)$. We have performed the $\hbar$-deformation in terms of insertion of $\hbar$ into the function $r(n)$. However, the prescription for deformation depends on the form of function.

\item We have revisited the case of KP hierarchy and observed that for some members of the hypergeometric family, in particular, Hurwitz numbers with completed cycles, there is a unique deformation prescription \eqref{hndef} that does not coincide with \eqref{KPhyp}. The reason for that difference is in the form (polynomial or exponential) of function $f(n)$, which carries the information about the spectral curve \cite{BDKS20}.

\item In the $B$KP case, we have observed that both Kontsevich and BGW models have the same simple pattern of $\hbar$-deformation while generating functions for spin Hurwitz numbers have the other one in accordance with \cite{AS21}. Reasons for the difference between BGW and spin Hurwitz numbers can be understood by analogy with the KP case. However, the case of the Kontsevich model suggests that the theory of \cite{BDKS20} can be further generalized to include $\tau$-functions of similar to the Kontsevich form.
\end{itemize}
Now let us discuss several questions which are related to the further research:
\begin{itemize}
    \item For both KP and $B$KP cases the theory of \cite{BDKS20,AS21} can be generalized to include $\tau$-functions of form similar to the Kontsevich model. That is, to the functions $r(n)$ defined up to $n \mod k$. For example, it would be useful to investigate the Hermitian matrix model with cubic potential and a proper star-like choice of contour.
    \item The recent increasing growth of interest in $W$-representations of matrix models \cite{MM22, WLZZ22, MMM23, MMM23a} revealed a large set of new matrix models. Firstly, it would be useful to make the genus expansion of these models and see how do they fit in the picture of $\hbar$-deformed hierarchies. Secondly, it would be interesting to insert $\hbar$ directly into the $W$-operators, which may uncover their additional structure.
    \item Finally, the same picture of $\hbar$-deformation should be extended to the $C$KP and $D$KP hierarchies, corresponding to the rest infinite-dimensional algebras. Even though these hierarchies are less studied and do not possess well-known matrix model solutions, their investigation is important in the search for the generalization to the $(q,t)$-KP hierarchy.
\end{itemize}

\section*{Acknowledgements}
We are grateful to Andrei Mironov, Alexei Morozov, Aleksandr Popolitov, and Alexey Sleptsov for careful reading of the paper and useful remarks. This work was supported by the Russian Science Foundation (Grant No.20-71-10073).
%

%
\printbibliography
\end{document}